\documentclass{llncs}

\usepackage[latin1]{inputenc}
\usepackage{subfigure}
\usepackage{enumerate}

\newcommand{\set}[1]{\ensuremath{\left\{#1\right\}}}

\usepackage{tikz}
\usepackage{pgfmath}
\usetikzlibrary{arrows}
\usetikzlibrary{snakes}
\usetikzlibrary{patterns}

\newcounter{tmpPolyominoSquareBool}
\newcounter{tmpPolyominoColumnBool}
\newcounter{tmpPolyominoColumnCounter}

\newcommand\hSquare[3]{\draw[fill=#1] (#2,#3) rectangle (#2+1,#3+1);}
\newcommand\hPattSquare[2]{\hSquare{white}{#1}{#2}
                           \draw[pattern=north east lines]
                             (#1,#2) rectangle (#1+1,#2+1);}

\newcommand\dominoTall[3][cyan]{%
  \draw[fill=#1] (#2,#3+1) rectangle (#2+1,#3+2+1);}
\newcommand\dominoLong[3][cyan]{%
  \draw[fill=#1] (#2,#3+1) rectangle (#2+2,#3+1+1);}

\newcommand\hSwapColorSquare[3]{
  \def\counter{\arabic{#3}}
  \if\counter0{\hSquare{black}{#1}{#2}}\addtocounter{#3}{1}
  \else{\hSquare{white}{#1}{#2}}\setcounter{#3}{0}\fi}

\newcommand{\hColumn}[3]{
  \def\arg{#2}
  \if\arg0\relax\else{
  \setcounter{tmpPolyominoSquareBool}{#3}
    \foreach \y in {1,...,#2}{
      \hSwapColorSquare{#1}{\y}{tmpPolyominoSquareBool}}}\fi}

\newcommand{\hManhattan}[2]{
  \setcounter{tmpPolyominoColumnBool}{0}
  \setcounter{tmpPolyominoColumnCounter}{0}

  \def\counter{\arabic{tmpPolyominoColumnBool}}

  \foreach \height in {#2}{
    \if\counter0{
      \hColumn{#1+\value{tmpPolyominoColumnCounter}}{\height}{0}
      \addtocounter{tmpPolyominoColumnBool}{1}}
    \else{
      \hColumn{#1+\value{tmpPolyominoColumnCounter}}{\height}{1}
      \setcounter{tmpPolyominoColumnBool}{0}}\fi
    
    \addtocounter{tmpPolyominoColumnCounter}{1}}}

\newcommand{\Manhattan}[2][]{
  \hManhattan{0}{#2}

  \def\arg{#1}\def\null{}
  \if\arg\null\relax\else{
  \setcounter{tmpPolyominoColumnCounter}{0}
  \foreach \lab in {#1}{
    \draw (\value{tmpPolyominoColumnCounter}+0.5,0.5) node { \small\lab } ;
    \addtocounter{tmpPolyominoColumnCounter}{1}}}\fi
}

\newcommand{\hSpacedColumns}[3]{
  \setcounter{tmpPolyominoColumnBool}{0}
  \setcounter{tmpPolyominoColumnCounter}{0}

  \def\counter{\arabic{tmpPolyominoColumnBool}}

  \foreach \height in {#2}{
    \if\counter0{
      \hColumn{#1+\value{tmpPolyominoColumnCounter}*#3}{\height}{0}
      \addtocounter{tmpPolyominoColumnBool}{1}}
    \else{
      \hColumn{#1+\value{tmpPolyominoColumnCounter}*#3}{\height}{1}
      \setcounter{tmpPolyominoColumnBool}{0}}\fi
    
    \if\height0\relax\else\addtocounter{tmpPolyominoColumnCounter}{1}\fi}}

\newcommand{\SpacedColumns}[3][]{
  \hSpacedColumns{0}{#2}{#3}

  \def\arg{#1}\def\null{}
  \if\arg\null\relax\else{
  \setcounter{tmpPolyominoColumnCounter}{0}
  \foreach \lab in {#1}{
    \draw (\value{tmpPolyominoColumnCounter}*#3+0.5,0.5) node { \small\lab } ;
    \addtocounter{tmpPolyominoColumnCounter}{1}}}\fi
}

\newcommand{\upperBound}[1]{\ensuremath{O\left(#1\right)}}

\usepackage{amsmath}
\usepackage{amssymb}

\author{Olivier Bodini\thanks{Supported by ANR contract GAMMA,
    ``G\'en\'eration Al\'eatoire Mod\`eles, M\'ethodes et Algorithmes'',
    BLAN07-2 195422. } \and J\'er\'emie Lumbroso}

\title{Optimal Partial
  Tiling of Manhattan~Polyominoes}
\institute{
LIP6, UMR 7606, CALSCI departement,
Universit\'e Paris 6 / UPMC,\\
104, avenue du pr\'esident Kennedy,\\ 
F-75252 Paris cedex 05, France\\
}
\date{\today}

\begin{document}

\maketitle
\begin{abstract}
  Finding an efficient optimal partial tiling algorithm is still an open
  problem. We have worked on a special case, the tiling of Manhattan
  polyominoes with dominoes, for which we give an algorithm linear in the
  number of columns. Some techniques are borrowed from traditional graph
  optimisation problems.
\end{abstract}

For our purpose, a \emph{polyomino} is the (non necessarily connected) union of unit squares (for which we
will consider the vertices to be in $\mathbb{Z}^2$) and a \emph{domino}
is a polyomino with two edge-adjacent unit squares.

To solve the domino tiling problem \cite{Fo,It,Ke,Re,Thi,Th} for
a polyomino $P$ is equivalent to finding a perfect matching in the
edge-adjacent graph $G_p$ of $P$'s unit squares. A specific point of view
to study tiling problems has been introduced by Conway and Lagarias in
\cite{CL}: they transformed the tiling problem into a combinatorial group
theory problem. Thurston \cite{Th} built on this idea by introducing the
notion of height, with which he devised a linear-time algorithm to tile
a polyomino without holes with dominoes. Continuing these works, Thiant
\cite{Thi}, and later Bodini and Fernique \cite{BF} respectively obtained
an \upperBound{n\log n} algorithm which finds a domino tiling for
a polyomino whose number of holes is bound, and an \upperBound{n\log^3 n}
algorithm with no constraint as to the number of holes. All these advances
involved the ``exact tiling problem''.

Naturally, all aforementioned exact tiling algorithms are useless when
confronted with polyominoes that cannot be perfectly tiled with dominos.
Such algorithms will output only that the polyominoes cannot be
tiled---and this is inconvenient because a perfect tiling is not
necessarily required: we might want partial tilings with at most $x$
untiled squares.

Thus, by analogy with matching problems in graph theory and the notion of
maximum matching, it seems interesting to study the corresponding notion
of optimal partial tiling. No algorithm has been specifically designed
with this purpose in mind; and whether the optimal partial tiling problem
for polyominoes can be solved by a linear-time algorithm has been an open
problem for the past 15 years. In this paper, we are studying the partial
optimal domino tiling problem on the Manhattan class of polyominoes (see
figure \ref{fig1}).

A \emph{Manhattan polyomino} is a sequence of adjacent columns that all
extend from the same base line (``\emph{Manhattan}'' refers to the fact
that this polyomino looks like a skyline). We are exhibiting an algorithm
that makes solving this problem much more efficient than solving the
corresponding maximum matching problem. Indeed, the best known algorithm
for maximum matchings in bipartite graphs is in \upperBound{\sqrt{n}m}
\cite{GT,HK} (where $n$ is the number of vertices and $m$ the number of
edges) and it yields an algorithm in \upperBound{n^{3/2}} for the partial
tiling problem, whereas our algorithm is linear in the number of columns. The choice of the Manhattan class of polyominoes is justified
as a generalisation of the seminal result of Boug\'e and Cosnard \cite{BC}
which establishes that a trapezoidal polyomino (i.e. a monotic Manhattan
polyomino, in which the columns are ordered by increasing height) is
tilable by dominoes if and only if it is balanced. The partial tiling
problem on more general classes seems to be out of reach at present.

\begin{figure}[htbp]
  \centering
  \begin{tikzpicture}[scale=0.6]
    \Manhattan{4,2,4,4,1,2,2,2,4,4}
  \end{tikzpicture}
  \caption{A Manhattan polyomino.}
  \label{fig1}
\end{figure}
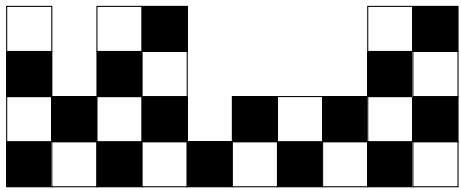

This paper first succinctly recalls some definitions on tilings,
polyominoes in general and Manhahattan polyominoes in particular; we list
the notations that we will use throughout the paper (section 1). Follows
a presentation of the main idea of reducing the partial tiling problem to
a network flow problem (section 2); we then prove this reduction to be
valid (section 3). A greedy algorithm to solve this specific network
problem is given (section 4), which allows us to conclude (section 5).

\section{Definitions}
\label{sec:defs}

An \emph{optimal partial tiling} by dominoes of a polyomino $P$ is a set
$\set{D_1, \ldots, D_k}$ of vertical or horizontal dominoes placed in the
discret plane such that:
\begin{enumerate}[(i)]
\item $\forall i \in \set{1, \ldots, k} D_i \subset P$;
\item the dominoes $D_1, \ldots, D_k$ have mutually disconnected
  interiors: for every $i, j$ such that $1 \leqslant i \leqslant
  j \leqslant k$, we have $\mathop{D_i}\limits^\circ \cap
  \mathop{D_j}\limits^\circ = \emptyset$ where $\mathop{D}\limits^\circ$
  is the interior of $D$ for the topology endowed with the euclidean
  distance on $\mathbb{R}^2$;
\item $k$ is maximal: if $k'$, $k' \geqslant k$, and $D'_1, \ldots D'_k$
  exists with the above two conditions, then $k = k'$.
\end{enumerate}
A \emph{column} is the union of unit squares placed one on top of the
other. A \emph{Manhattan polyomino} is the connected union of columns
aligned on the same horizontal line.

We will identify a Manhattan polyomino with a tuple: the sequence of the
heights of all the columns of the polyomino (ordered from left to right).
For example, the polyomino in figure \ref{fig1} is represented as $(4, 2,
4, 4, 1, 2, 2, 2, 4, 4)$. This representation is clearly more convenient
in memory space than, say, the list of all unit squares of a polyomino.

It is convenient to define \emph{inclusion}, and we will do so using the
tuple representation: polyomino $P=\left(p_1, p_2, \ldots, p_n\right)$ is
said to be included in $Q=\left(q_1, q_2, \ldots, q_n\right)$ if $\forall
i, p_i \leqslant q_i$ (both polyominoes can be padded with zero-sized
columns if necessary).

In addition, the unit squares of a polyomino have a chessboard-like
coloration; specifically\footnote{And regardless of the basis, as both
  colors are symmetrical in purpose (i.e.: it doesn't matter whether
  a given unit square is white or black, so long as the whole polyomino is
  colored in alternating colors).}, the unit square with position $(i, j)$
(i.e.: unit square $(i, j) + [0,1[^2$) is \emph{white} if $i+j$ is even
and \emph{black} otherwise. A column is said to be \emph{white dominant}
(or \emph{black dominant}) if it has more white (or black) unit squares
than black (or white) ones.

\section{From polyominoes to flow networks}
\label{sec:to-flow}

In this section, without loss of generality, we consider the polyominoes
are balanced (a polyomino is \emph{balanced} if it contains the same count
of white and black unit squares). Indeed, a polyomino can always be
balanced by adding the necessary number of isolated unit squares of the
lacking color.

\paragraph{Notations.}

Let $(h_1, \ldots, h_n)$ be a Manhattan polyomino $P$.

For $i \leqslant j$, we define $X(i, j)$ as the number $\min_{i\leqslant
  k \leqslant j} \set{h_k}$, the height of the smallest column contained in
the subset $\set{h_i, \ldots, h_j}$.

We define $B(P)$ (and respectively $W(P)$) as the set of black unit
squares, (or white unit squares) contained in $P$. Let $I_P$ (and resp.
$J_P$) be the set of indices of black dominant columns (and resp. white
dominant columns). We define $s_1, \ldots, s_n$ as the elements of $I_P
\cup J_P$ sorted in ascending order.

Let $G=(V,E)$ be a graph, and $S \subset V$. We define $\Gamma(S)$ as the
subset of \emph{neighbor vertices} of $S$, i.e. the vertices $y \in V$
such that there exists $x \in S$ with $(x, y) \in E$ or $(y, x) \in E$.

\paragraph{Construction of the flow network.}

For each polyomino $P$, we build a directed graph (\emph{flow
  network}) that we call $F_P$:
\begin{itemize}
\item its vertex set is $\set{s_1, \ldots, s_n} \cup \set{s, t}$ where $s$
  and $t$ are two additional vertices respectively called \emph{source}
  and \emph{sink} (be mindful of the fact that we use $s_i$ both to refer
  to the vertices, and to the actual columns);
\item its edge set $E$ is defined by:
  \begin{align*}
    E & {} = \set{\left(s_i, s_{i+1}\right) ; i \in \set{1, \ldots, n-1}} \\
    & {} \cup \set{\left(s_{i+1}, s_i\right) ; i \in \set{1, \ldots, n-1}} \\
    &{} \cup \set{(s, a) ; a \in I_P} \cup \set{(a, t) ; a \in J_P}\text{,}
  \end{align*}
  and each arc of $E$ is weighted by a function $c$ called
  \emph{capacity} function defined by:
  \[
    c(e) = \left\{
      \begin{array}{ccl}
        \left\lceil\frac{X\left(s_i, s_{i+1}\right)}{2}\right\rceil&\qquad
        &\text{if $e = \left(s_i, s_{i+1}\right)$ or
          $e = \left(s_{i+1}, s_i\right)$}\\
        1&\qquad&\text{if $e=(s, a)$ with $a \in I_P$}\\
        1&\qquad&\text{if $e=(b, t)$ with $b \in J_P$}
      \end{array}\right.
  \]
  where $x \mapsto \lceil x\rceil$ is the ceiling function (which returns
  the smallest integer not lesser than $x$).
\end{itemize}

%
%
%

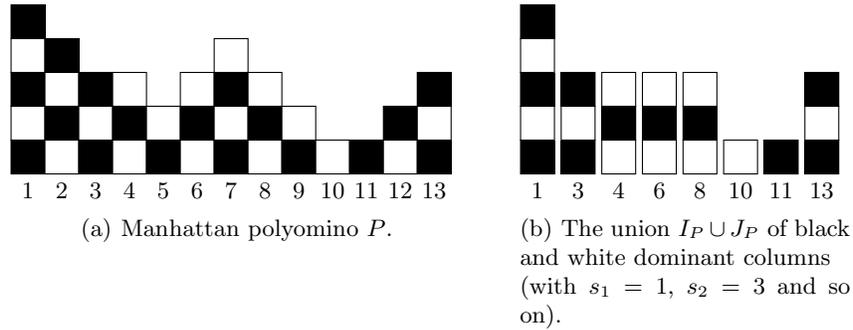
\begin{figure}[htbp]
  \begin{center}
    \subfigure[Manhattan polyomino $P$.\label{fig2a}]{\begin{tikzpicture}[scale=0.45]
      \Manhattan[1,2,3,4,5,6,7,8,9,10,11,12,13]{5,4,3,3,2,3,4,3,2,1,1,2,3}
    \end{tikzpicture}
}
     \qquad%
    \subfigure[The union $I_P \cup J_P$ of black and white dominant
  columns\newline (with $s_1 = 1$, $s_2 = 3$ and so on).
    \label{fig2b}]{\begin{tikzpicture}[scale=0.45]
      \SpacedColumns[1,3,4,6,8,10,11,13]{5,0,3,3,0,3,0,3,0,1,1,0,3}{1.2}
    \end{tikzpicture}
}
    \caption{Construction of a $\set{s_1, \ldots, s_n}$ set.}
  \label{fig2}
  \end{center}
\end{figure}

\paragraph{Example.}

Figures \ref{fig2} and \ref{fig3} illustrate the construction described
above. The polyomino $P$ in figure \ref{fig2a} contains four black
dominant columns and four white dominant columns which have been extracted
and are represented on figure \ref{fig2b}.

This yields the set $I_P \cup J_P = \set{s_1 = 1, s_2 = 3, \ldots, s_8 =
  13}$ from which we can build the flow network in figure \ref{fig3}.

\begin{figure}[htbp]
  \centering
  \begin{tikzpicture}[scale=0.7]
    \tikzstyle{arc}=[->, >=triangle 45, very thin]
    \tikzstyle{inside arc}=[->, >=triangle 45, very thin, bend right=45]
    \tikzstyle{bc}=[shape=circle, draw, fill=black]
    \tikzstyle{wc}=[shape=circle, draw, fill=white]
    \tikzstyle{sink}=[shape=circle, draw, fill=black!3]
    \tikzstyle{lbl}=[auto, fontsize=7pt]

    \path (3.5*1.3,3)  node[sink] (t) {\small t}
          (3.5*1.3,-3) node[sink] (s) {\small s};

    \path (0*1.3,0) node[bc] (1) {}
          (1*1.3,0) node[bc] (2) {}
          (2*1.3,0) node[wc] (3) {}
          (3*1.3,0) node[wc] (4) {}
          (4*1.3,0) node[wc] (5) {}
          (5*1.3,0) node[wc] (6) {}
          (6*1.3,0) node[bc] (7) {}
          (7*1.3,0) node[bc] (8) {};

    \draw[inside arc] (1) to node[swap,auto] {\footnotesize 2} (2);
    \draw[inside arc] (2) to node[swap,auto] {\footnotesize 2} (3);
    \draw[inside arc] (3) to node[swap,auto] {\footnotesize 1} (4);
    \draw[inside arc] (4) to node[swap,auto] {\footnotesize 2} (5);
    \draw[inside arc] (5) to node[swap,auto] {\footnotesize 1} (6);
    \draw[inside arc] (6) to node[swap,auto] {\footnotesize 1} (7);
    \draw[inside arc] (7) to node[swap,auto] {\footnotesize 1} (8);

    \draw[inside arc] (2) to node[swap,auto] {\footnotesize 2} (1);
    \draw[inside arc] (3) to node[swap,auto] {\footnotesize 2} (2);
    \draw[inside arc] (4) to node[swap,auto] {\footnotesize 1} (3);
    \draw[inside arc] (5) to node[swap,auto] {\footnotesize 2} (4);
    \draw[inside arc] (6) to node[swap,auto] {\footnotesize 1} (5);
    \draw[inside arc] (7) to node[swap,auto] {\footnotesize 1} (6);
    \draw[inside arc] (8) to node[swap,auto] {\footnotesize 1} (7);

    \draw[arc, bend left=45] (s) to node[auto] {\footnotesize 1} (1);
    \draw[arc, bend left=35] (s) to node[swap,auto] {\footnotesize 1} (2);
    \draw[arc, bend right=35] (s) to node[auto] {\footnotesize 1} (7);
    \draw[arc, bend right=45] (s) to node[swap,auto] {\footnotesize 1} (8);

    \draw[arc, bend left=45] (3) to node[auto] {\footnotesize 1} (t);
    \draw[arc, bend left=10] (4) to node[auto] {\footnotesize 1} (t);
    \draw[arc, bend right=10] (5) to node[swap,auto] {\footnotesize 1} (t);
    \draw[arc, bend right=45] (6) to node[swap,auto] {\footnotesize 1} (t);
  \end{tikzpicture}

  \caption{Flow network associated with the polyomino of figure \ref{fig2a}.}
  \label{fig3}
\end{figure}
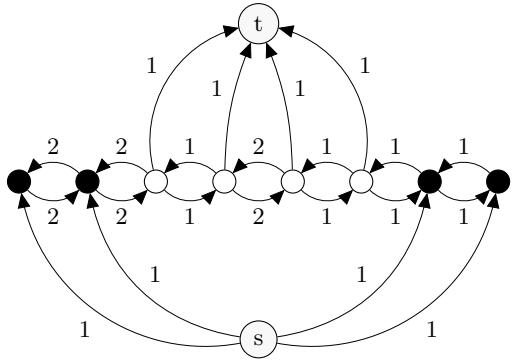

For instance, the capacity of arc $(s_1, s_2)$ is $c\left((s_1,
  s_2)\right) = 2$ because the smallest column between the indices $s_1 =
1$ and $s_2 = 3$ is the column $3$ and its height is $3$ (what this boils
down to is $X(1, 3) = 3$).

\section{Main theorem: An overview}

As might have been evident in the previous sections, oddly-sized columns,
whether white or black dominant, are what make domino tiling of Manhattan
polyominoes an interesting problem: indeed, evenly-sized columns (of
a polyomino comprising only such columns) can conveniently be tiled with
half their size in dominoes.

\subsection{Sketch of the method}

Let $P$ be a polyomino. Our method can intuitively be summarized in
several steps:

\begin{itemize}
\item we devise a ``\emph{planing}'' transformation which partially tiles
  the polyomino $P$ in a locally optimal way (i.e.: such that it is
  possible to obtain an optimal partial tiling for $P$ by optimally tiling
  the squares which remain uncovered after the transformation); this
  transformation, as illustrated in figure \ref{fig4c}, ``\emph{evens
    out}'' a black dominant column and a white dominant column;

\item we devise a way to construct a maximum flow problem that translates
  exactly how the planing transformation can be applied to the polyomino
  $P$;

\item we prove that the solutions to the maximum flow problem constructed
  from $P$, and to the optimal partial tiling of $P$ are closely linked
  (by a formula).
\end{itemize}

In truth, these steps are intertwined: the local optimality of the planing
transformation is demonstrated using the maximum flow problem.
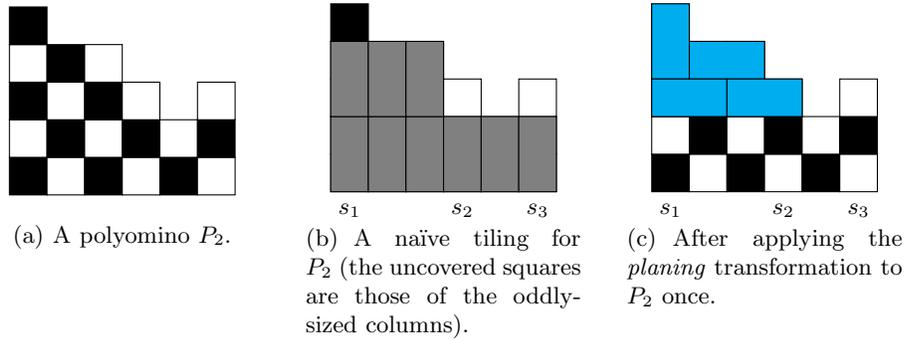
\begin{figure}[htbp]
 \begin{center}
    \subfigure[A polyomino $P_2$.\label{fig4a}]
    {\begin{minipage}{0.3\linewidth}\centering
    \begin{tikzpicture}[scale=0.5]
      \Manhattan[\quad]{5,4,4,3,2,3}
    \end{tikzpicture}\end{minipage}}%
\hfill%
 \subfigure[A na\"ive tiling for $P_2$ (the uncovered squares are those of
  the oddly-sized columns).\label{fig4b}]
{\begin{minipage}{0.3\linewidth}\centering
    \begin{tikzpicture}[scale=0.5]
      \Manhattan[$s_1$,,,$s_2$,,$s_3$]{5,4,4,3,2,3}
      \foreach \i in {0,...,5}{\dominoTall[gray]{\i}{0}}
      \foreach \i in {0,...,2}{\dominoTall[gray]{\i}{2}}
    \end{tikzpicture}\end{minipage}}%
  \hfill%
  \subfigure[After applying the \emph{planing} transformation to $P_2$
  once. \label{fig4c}]{\begin{minipage}{0.3\linewidth}\centering
    \begin{tikzpicture}[scale=0.5]
      \Manhattan[$s_1$,,,$s_2$,,$s_3$]{5,4,4,3,2,3}
      \dominoTall{0}{3}
      \dominoLong{1}{3}\dominoLong{0}{2}\dominoLong{2}{2}
    \end{tikzpicture}\end{minipage}}

  \caption{Tilings of $P_2$, which has three oddly-sized columns.}
  \label{fig4}
\end{center}
\end{figure}

\subsection{The planing (or leveling) transformation}

The planing transformation we just mentionned can be expressed as a map
$\phi$, which can only be applied to a polyomino $P$ which has two
\emph{consecutive}\footnote{By this we mean that the two dominant columns
  are consecutive in the $s_1, s_2, \ldots$ ordering (and this does not
  necessarily imply that they are adjacent in $P$); in figure \ref{fig4c},
  dominant columns $s_1$ and $s_2$ are consecutive (so are $s_2$ and
  $s_3$).} oddly-sized columns $s_i, s_{i+1}$ of \emph{different} dominant
colors. This condition\footnote{Wherein it is obvious that the \textbf{or}
  is mutually exclusive.} can be summarized as
\begin{align*}
  \exists i, \left(s_i \in I_P \text{ and } s_{i+1} \in J_P
    \quad\text{\textbf{or}}\quad s_i \in J_P \text{ and } s_{i+1} \in
    I_P\right)\text{.}
\end{align*}

Let $P=\left(h_1, \ldots h_n\right)$ be a Manhattan polyomino, and let $i$
be the smallest integer such that $s_i, s_{i+1}$ are two oddly-sized
columns of different dominant colors and $c(s_i,s_{i+1})\neq 0$, then $P' = \phi(P)$ is a new
polyomino defined by
\begin{align*}
  P' = \left(h_1, \ldots, h_{s_i - 1}, a, \ldots, a, h_{s_{i+1} +1},
    \ldots, h_n\right)
\end{align*}
where $a = \min\left(h_{s_i} - 1, h_{s_i+1} - 2, \ldots,
  h_{s_{i+1} -1} -2, h_{s_{i+1}}-1\right)$, i.e.: all columns from $s_i$ to
$s_{i+1}$ are \emph{leveled} to height $a$, which is at least one unit
square smaller than the smallest of $s_i$ and $s_{i+1}$
 and two unit
squares smaller than the smallest even-sized column of $P$ from $s_i$ to $s_{i+1}$
.

Since we have only decreased the heights of the columns of $P$ to obtain
$P'$, $\phi(P) = P' \subset P$ is trivial (following the definition of
inclusion we gave in section \ref{sec:defs}).

\begin{figure}[htbp]
 \begin{center}
 \subfigure[At least one unit square smaller than the smallest odd-sized
  column ...\label{fig5a}]{
    \begin{tikzpicture}[scale=0.4]
      \Manhattan[$s_1$,,,,,$s_2$,,,]{3,6,8,6,6,5,4,2,2}
      \foreach \i in {0,2,...,4}{\dominoLong{\i}{2}}
      \dominoLong{1}{3}\dominoLong{3}{3}\dominoTall{5}{3}
      \foreach \i in {1,...,4}{\dominoTall{\i}{4}}
      \dominoTall{2}{6}
    \end{tikzpicture}
  }%
  \qquad%
  \subfigure[... at least two unit squares smaller than the smallest
  even-sized column. \label{fig5b}]{
    \begin{tikzpicture}[scale=0.4]
      \Manhattan[$\,$,,,,$s_1$,,,$s_2$,,,$s_3$,,$s_4$]{%
        2,2,4,4,7,6,4,5,4,4,3,2,3}
      \dominoLong{4}{2}\dominoLong{6}{2}
      \dominoTall{4}{3}\dominoLong{5}{3}\dominoTall{7}{3}
      \dominoTall{5}{4}
      \dominoTall{4}{5}
    \end{tikzpicture}
  }

  \caption{Generic examples of the planing transformation.}
  \label{fig5}
\end{center}

\end{figure}

We will show that this transformation is locally optimal: again, this
means that if optimal tilings of $\phi(P)$ have $x$ non-covered unit
squares, then so do optimal tilings of $P$ (the number of non-covered unit
squares is invariant with regards to the transformation).

\subsection{Main theorem}

Having constructed a flow network $F_P$ from a given Manhattan polyomino
$P$, we now show that, if we call $d(P)$ the number of non-covered unit
squares in an optimal partial tiling of $P$, then the value $v(P)$ of the
maximum flow on $F_P$ is such that:
\[
\left|I_P\right| + \left|J_P\right| - 2v(P) = d(P)\text{.}
\]
The meaning of this equation is quite intuitive: it says that the number
of uncovered squares $d(P)$, is exactly the number of single squares that
come both from the black dominant columns $I_P$ and from the white
dominant columns $J_P$, from which we withdraw two squares for each of the
$v(P)$ times we are able to apply the planing transformation (recall that
this transformation evens out both a black dominant column and a white
dominant column, hence the $2$ factor).

\begin{lemma}\label{lemma:inv}
  Let $P$ be a Manhattan polyomino. For all $P'$, such that $P' =
  \phi^k(P)$ (the planing transformation applied $k$ times to $P$), the
  following invariant holds:
  \[
  \left|I_P\right| + \left|J_P\right| - 2v(P) = \left|I_{P'}\right| +
  \left|J_{P'}\right| - 2v(P')
  \]
  where $v(P)$ and $v(P')$ are the values of a maximum flow respectively
  in $T_P$ and $T_{P'}$.
\end{lemma}

\begin{proof}[by induction on $k$] Let $f$ be a maximum flow on
  $T_{P_{k-1}}$ which minimizes $\sum_e f(e)$ (i.e.: the flow takes the
  shortest path from $s$ to $t$ given a choice). Since it is possible to
  apply the planing transformation to $P_{k-1}$, then there must be two
  consecutive vertices $s_i$ and $s_{i+1}$ corresponding to oddly-sized
  columns of different dominant colors (given several such pairs, we
  consider that for which $i$ is the smallest).

  Let us now build $T_{P_{k}}$, the flow network associated with $P_k =
  \phi(P_{k-1})$~: the vertices of $T_{P_k}$ are the vertices of
  $T_{P_{k-1}}$ from which we remove $\left\{ s_i, s_{i+1}\right\}$; the
  arcs of $T_{P_K}$ are the trace of the arcs of the arcs of $T_{P_{k-1}}$
  (\emph{trace} is a non-standard term by which we mean all arcs that are
  incident neither to $s_i$ nor to $s_{i+1}$) to which we add, should both
  vertices exist, the arcs $(s_{i-1}, s_{i+2})$ and $(s_{i+2}, s_{i-1})$
  with capacity
  \[
  \min\left( {c\left( {\left( {s_{i - 1} ,s_i } \right)}
      \right),c\left( {\left( {s_i ,s_{i + 1} } \right)} \right) -
      1,c\left( {\left( {s_{i + 1} ,s_{i + 2} } \right)} \right)}
  \right)\text{.}
  \]
  Suppose $s_i$ and $s_{i+1}$ respectively represent a black and white
  dominant column (figure \ref{fig6} illustrates the construction in this
  case); the other case is symmetrical. Then, because $\sum_e f(e)$ is
  minimal, maximum flow $f$ routes a non-null amount through the path $(s,
  s_i, s_{i+1}, t)$.

  \begin{figure}[htbp]
    \centering
    \def\captionsize{\large}
    \begin{tikzpicture}[scale=0.5, transform shape]
      \tikzstyle{arc}=[->, >=triangle 45, very thin]
      \tikzstyle{inside arc}=[->, >=triangle 45, very thin, bend right=45]
      \tikzstyle{bc}=[shape=circle, draw, fill=black]
      \tikzstyle{wc}=[shape=circle, draw, fill=white]
      \tikzstyle{sink}=[shape=circle, draw, fill=black!3]
      \tikzstyle{lbl}=[auto, fontsize=7pt]

      \path (3.5*1.3,3)  node[sink] (t) {\captionsize t}
            (3.5*1.3,-3) node[sink] (s) {\captionsize s};

      \path (0*1.3,0) node[bc] (1) {}
            (1*1.3,0) node[bc] (2) {}
            (2*1.3,0) node[wc] (3) {}
            (3*1.3,0) node[wc] (4) {}
            (4*1.3,0) node     (5) {};

      \draw[inside arc, dotted] (4) to (5);
      \draw[inside arc, dotted] (5) to (4);

      \draw[inside arc] (1) to node[swap,auto] {\captionsize 2} (2);
      \draw[inside arc] (2) to node[swap,auto] {\captionsize 2} (3);
      \draw[inside arc] (3) to node[swap,auto] {\captionsize 1} (4);

      \draw[inside arc] (2) to node[swap,auto] {\captionsize 2} (1);
      \draw[inside arc] (3) to node[swap,auto] {\captionsize 2} (2);
      \draw[inside arc] (4) to node[swap,auto] {\captionsize 1} (3);
      
      \draw[arc, bend left=45] (s) to node[auto] {\captionsize 1} (1);
      \draw[arc, bend left=35] (s) to node[swap,auto] {\captionsize 1} (2);

      \draw[arc, bend left=45] (3) to node[auto] {\captionsize 1} (t);
      \draw[arc, bend left=10] (4) to node[auto] {\captionsize 1} (t);
    \end{tikzpicture}\hfill
    \begin{tikzpicture}[scale=0.5, transform shape]
      \tikzstyle{sp arc}=[->, >=triangle 45, very thin]
      \tikzstyle{sp inside arc}=[->, >=triangle 45,very thin, bend right=45]
      \tikzstyle{arc}=[sp arc, gray, dashed]
      \tikzstyle{inside arc}=[sp inside arc, gray, dashed]
      \tikzstyle{bc}=[shape=circle, draw, fill=black]
      \tikzstyle{wc}=[shape=circle, draw, fill=white]
      \tikzstyle{sink}=[shape=circle, draw, fill=black!3]
      \tikzstyle{lbl}=[auto, fontsize=7pt]

      \path (3.5*1.3,3)  node[sink] (t) {\captionsize t}
            (3.5*1.3,-3) node[sink] (s) {\captionsize s};

      \path (0*1.3,0) node[bc] (1) {}
            (1*1.3,0) node[bc] (2) {}
            (2*1.3,0) node[wc] (3) {}
            (3*1.3,0) node[wc] (4) {}
            (4*1.3,0) node     (5) {};

      \draw[inside arc, dotted] (4) to (5);
      \draw[inside arc, dotted] (5) to (4);

      \draw[inside arc] (1) to node[swap,auto] {} (2);
      \draw[sp inside arc, thick]
                (2) to node[swap,auto] {\captionsize \;\;\;\; 2 - 1} (3);
      \draw[inside arc] (3) to node[swap,auto] {} (4);

      \draw[inside arc] (2) to node[swap,auto] {} (1);
      \draw[inside arc] (3) to node[swap,auto] {} (2);
      \draw[inside arc] (4) to node[swap,auto] {} (3);
      
      \draw[arc, bend left=45] (s) to node[auto] {} (1);
      \draw[sp arc, thick, bend left=35]
                      (s) to node[swap,auto] {\captionsize 1 - 1} (2);

      \draw[sp arc, thick, bend left=45] (3)
             to node[auto] {\captionsize 1 - 1} (t);
      \draw[arc, bend left=10] (4) to node[auto] {} (t);
    \end{tikzpicture}\hfill
    \begin{tikzpicture}[scale=0.5, transform shape]
      \tikzstyle{arc}=[->, >=triangle 45, very thin]
      \tikzstyle{inside arc}=[->, >=triangle 45, very thin, bend right=30]
      \tikzstyle{bc}=[shape=circle, draw, fill=black]
      \tikzstyle{wc}=[shape=circle, draw, fill=white]
      \tikzstyle{sink}=[shape=circle, draw, fill=black!3]
      \tikzstyle{lbl}=[auto, fontsize=7pt]

      \path (3.5*1.3,3)  node[sink] (t) {\captionsize t}
            (3.5*1.3,-3) node[sink] (s) {\captionsize s};

      \path (0*1.3,0) node[bc] (1) {}
            (3*1.3,0) node[wc] (4) {}
            (4*1.3,0) node     (5) {};

      \draw[inside arc, dotted] (4) to (5);
      \draw[inside arc, dotted] (5) to (4);
      \draw[arc, bend left=45] (s) to node[auto] {\captionsize 1} (1);
      \draw[arc, bend left=10] (4) to node[auto] {\captionsize 1} (t);

      \draw[inside arc] (1) to node[swap,auto] {
         \captionsize \;$\min(2,2\mbox{-}1,1) = 1$} (4);
      
      \draw[inside arc] (4) to node[swap,auto] {
         \captionsize $\min(2,2,1)=1$} (1);
      
    \end{tikzpicture}

    \caption{Building $T_{P_k}$ from $T_{P_{k-1}}$: a max flow for the
      network of figure \ref{fig3} necessarily goes through the path $(s,
      s_2, s_3, t)$. Here we show the impact of removing this path on the
      flow network's edges, vertices and capacities.}
    \label{fig6}
  \end{figure}
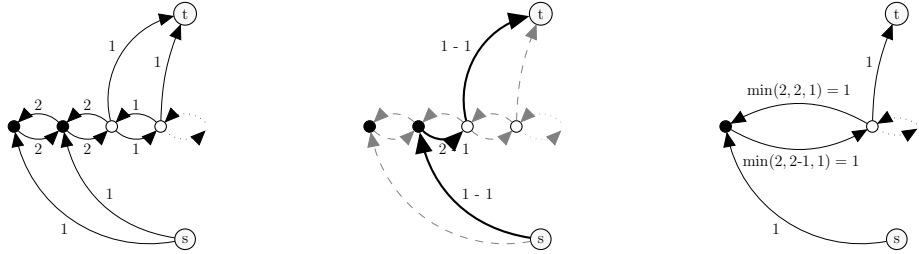

  Thus the maximum flow $f$ on $T_{P_{k-1}}$ induces a maximum flow $f'$
  on $T_{P_k}$, which is defined as the trace of $f$ on $T_{P_k}$ and
  updated, if necessary, by adding:
  \[\left\{
      \begin{array}{c}
        {f'}((s_{i-1},s_{i+2}))=f((s_{i-1},s_{i}))=f((s_{i+1},s_{i+2}))\\
        {f'}((s_{i+2},s_{i-1}))=f((s_{i},s_{i-1}))=f((s_{i+2},s_{i+1}))
      \end{array}\right..
  \]
  Proving that $f'$ is a maximum flow on $T_{P_k}$ is straightforward
  (suppose there is a better flow, and reach a contradiction). We then
  immediately remark that $v(P_{k-1}) = v(P_k) + 1$; and because the
  planing transformation applied to $P_{k-1}$ levels both a black dominant
  column and a white dominant column to obtain $P_k$, we have
  $\left|I_{P_k}\right| = \left|I_{P_{k-1}}\right| - 1$ and
  $\left|J_{P_k}\right| = \left|J_{P_{k-1}}\right| - 1$, thus:
  \[
  \left|I_{P_{k-1}}\right| + \left|J_{P_{k-1}}\right| - 2v(P_{k-1}) =
  \left|I_{P_k}\right| + \left|J_{P_k}\right| - 2v(P_k)\text{,}
  \]
  from which, by induction on $k$, we derive the lemma's invariant.\qed
\end{proof}

\begin{lemma}\label{lemma:cons}
  Let $P$ be a Manhattan polyomino and $v(P)$ the value of a maximum flow
  in $T_P$ (cf.~section~\ref{sec:to-flow}). There is a partial tiling of
  $P$ in which exactly $\left|I_P\right| + \left|J_P\right| - 2v(P)$
  squares are uncovered.
\end{lemma}

\begin{proof}
  We prove the property by induction on the size of $P$. Let $f$ be
  a maximum flow in $T_P$. We consider whether there are two consecutive
  vertices $s_i$ and $s_{i+1}$ which correspond to oddly-sized columns of
  different dominant colors, such that the flow of $f$ going from $s_i$ to
  $s_{i+1}$ is non-null.

  Either: no such vertices $s_i$ and $s_{i+1}$ exist\footnote{This case
    arises when there are no oddly-sized columns, or only oddly-sized
    columns of a given color.}, in which case we have $v(P) = 0$ (this
  would mean $f$ is the \emph{zero flow}). Indeed, any path from $s$ to
  $t$ would use at least one transversal arc $(s_i, s_{i+1})$ because, by
  construction, a vertex $s_i$ cannot be simultaneously linked to $s$ and
  $t$. We remark that the planing transformation cannot be applied. If we
  partially tile $P$, as in figure \ref{fig4b}, with vertical dominoes
  alone, leaving one square uncovered per oddly-sized column, we obtain
  a partial tiling which globally leaves $\left|I_P\right| +
  \left|J_P\right|$ squares uncovered---thus proving our property for this
  case.

  Or: such vertices $s_i$ and $s_{i+1}$ exist. We can then apply the
  planing transformation $\phi$, $P' = \phi(P)$. By induction, the
  property holds for $P'$. As we've seen, the (union of the) squares
  removed by the transformation is a domino-tileable Manhattan polyomino
  (figure \ref{fig5}); and by lemma \ref{lemma:inv},
  \begin{align*}
    \left|I_{P'}\right| + \left|J_{P'}\right| - 2v({P'}) =
    \left|I_P\right| + \left|J_P\right| - 2v(P)\text{.}
  \end{align*}
  \qed
\end{proof}

\begin{theorem}\label{thm:opt}
  Let $P$ be a Manhattan polyomino, and $d(P)$, the number of uncovered
  unit squares in an optimal partial tiling of $P$. The construction
  outlined in the proof of lemma \ref{lemma:cons} is optimal, i.e.:
  \[
  \left|I_P\right| + \left|J_P\right| - 2v(P) = d(p)\text{.}
  \]
\end{theorem}


\begin{remark}
  We are now going to work on $G_P$ which is the edge-adjacency graph of
  polyomino $P$: the vertices of $G_P$ are the unit squares of $P$, and an
  edge connects two vertices of $G_P$ if and only if the two corresponding
  squares share an edge in $P$. (By contrast, $T_P$ is the flow network
  constructed in section \ref{sec:to-flow}.)
\end{remark}

The idea behind the proof of theorem \ref{thm:opt} is as follow: we
isolate a subset ${\mathcal B}(P)$ of black squares of $P$, such that the
(white) neighbors (in graph $G_P$) to squares of ${\mathcal B}(P)$ verify
the relation
\[
\left|{\mathcal B}(P)\right| - \left|\Gamma\left({\mathcal
      B}(P)\right)\right| = \left|I_P\right| + \left|J_P\right| -
2v(P)\text{.}
\]
We can then conclude, using Hall's theorem, that $d(P) \geqslant
\left|{\mathcal B}(P)\right| - \left|\Gamma\left({\mathcal
      B}(P)\right)\right|$. To construct ${\mathcal B}(P)$, we use a
minimal cut\footnote{Recall that a cut of graph $G=(V,E)$ is a partition
  of its vertices in two subsets $X_1$ and $X_2$; let $E_C$ be the set of
  edges such that one vertex is in $X_1$ and the other in $X_2$: the value
  of the cut is the sum of the capacity of each edge in $E_C$; a minimal
  cut is a cut which minimizes this sum. \textbf{In the case of flow
    networks}, a cut separates the source $s$ from the sink $t$.} of graph
$T_P$, from which we deduce a list of ``\emph{bottlenecks}''. These
bottlenecks mark the boundary of zones containing squares of ${\mathcal
  B}(P)$.

Intuitively, once bottlenecks are planed (that is to say when we have
tiled them following the template given by our planing transformation),
they isolate zones which have either too many white or too many black
squares.

\begin{figure}[ht]
  \centering
  \begin{tikzpicture}[scale=0.4]
    \tikzstyle{bc}=[shape=circle, draw=white, fill=black]
    \tikzstyle{wc}=[shape=circle, draw=black, fill=white]

    \Manhattan{5,4,3,2,1,1,2,3,2,3,1,2,1,3,2,1,3,1}
    \hPattSquare{0}{1}\hPattSquare{0}{3}\hPattSquare{0}{5}
    \hPattSquare{1}{2}\hPattSquare{1}{4}
    \hPattSquare{2}{1}\hPattSquare{2}{3}
    \hPattSquare{3}{2}\hPattSquare{4}{1}

    \path (9.0,8.0) node[draw, shape=circle] (t) {t}
          (7.0,-3.0) node[draw, shape=circle] (s) {s};

    \path (0.5,0.5) node[bc] (1) {}
          (2.5,0.5) node[bc] (2) {}
          (4.5,0.5) node[bc] (3) {}
          (5.5,2.5) node[wc] (4) {}
          (7.5,4.5) node[wc] (5) {}
          (9.5,4.5) node[wc] (6) {}
          (10.5,0.5) node[bc] (7) {}
          (12.5,0.5) node[bc] (8) {}
          (13.5,4.5) node[wc] (9) {}
          (15.5,2.5) node[wc] (10) {}
          (16.5,0.5) node[bc] (11) {}
          (17.5,2.5) node[wc] (12) {};

    \tikzstyle{flow}=[decorate,decoration={snake, post length=2mm}]
    \tikzstyle{arc}=[->, >=triangle 45, very thin]
    \draw[arc, bend left=20] (s) to (1);
    \draw[arc, bend left=20] (s) to (2);
    \draw[arc, bend left=20, flow] (s) to (3);
    \draw[arc, bend right=30, flow] (s) to (7);
    \draw[arc, bend right=25, flow] (s) to (8);
    \draw[arc, bend right=20, flow] (s) to (11);

    \draw[arc, bend left=40, flow] (4) to (t);
    \draw[arc, bend left=20] (5) to (t);
    \draw[arc, bend right=10, flow] (6) to (t);
    \draw[arc, bend right=20, flow] (9) to (t);
    \draw[arc, bend right=40, flow] (10) to (t);
    \draw[arc, bend right=50] (12) to (t);

    \path[draw, dashed, gray] (1) -- (2) -- (3)
      (4) -- (5) -- (6) (7) -- (8) (9) -- (10)
      (11) -- (12);

    \draw[arc, flow, white, very thick] (3) to (4);
    \draw[arc, flow, white, very thick] (7) to (6);
    \draw[arc, flow, white, very thick] (8) to (9);
    \draw[arc, flow, white, very thick] (11) to (10);

    \draw[arc, flow] (3) to (4);
    \draw[arc, flow] (7) to (6);
    \draw[arc, flow] (8) to (9);
    \draw[arc, flow] (11) to (10);

    \draw[red, very thick] (3,5) to (9, -4);
  \end{tikzpicture}
  \label{fig7}
  \caption{A polyomino $P$, its flow network $T_P$ (the transversal arcs
    are only hinted at, with dashed lines; the max flow is represented
    with snake-lines), and, in red, a minimal cut of $T_P$ which
    separates $P$ in two zones (each zone is connected here, but that
    may not be the case). The black striped squares on the left belong
    to ${\mathcal B}(P)$: in accordance with Hall's theorem,
    $\left|{\mathcal B}(P)\right| - \left|\Gamma({\mathcal B}(P))\right| =
    9 - 7 = 2$ is the number of uncovered squares in the optimal partial
    tiling.}
\end{figure}
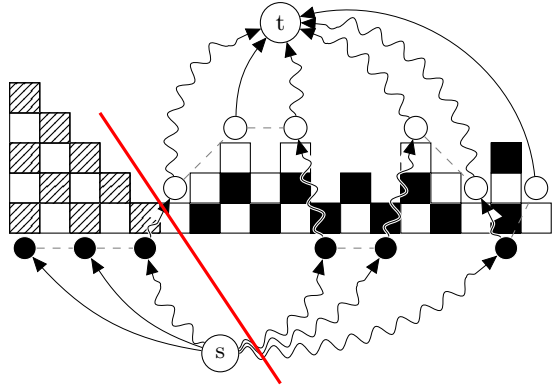

\subsection{Greedy algorithm}
Having proven that our optimal partial tiling problem can be reduced to
a network flow problem, we will now present an algorithm to efficiently
solve this specific brand of network flow problem (indeed,
a general-purpose network flow algorithm would not take into account those
properties of our network which are a consequence of the way it was
constructed).

We will first introduce the notion of $f$-tractability for two vertices of
our flow network. The pair $\left(s_i, s_j\right)$ (for which we consider
that $i < j$) is said to be $f$-tractable if:
\begin{itemize}
\item $s_i$ and $s_j$ correspond to two oddly-sized columns of
  \emph{different} dominant columns\footnote{Recall that this is
    \emph{almost} the condition under which the plan\-ing transformation
    can be applied: the planing transformation requires, in addition, that
    $s_i$ and $s_j$ be so that $j = i+1$ ---per the ordering
    which have defined in section \ref{sec:to-flow}.};
\item if $s_i$ corresponds to a black-dominant column, then the arcs
  $\left(s_i, s_{i+1}\right)$, ..., $\left(s_{j-1}, s_j\right)$ are not
  saturated;
\item or else, if $s_i$ corresponds to a white-dominant column, then the
  arcs $\left(s_j, s_{j-1}\right)$, ..., $\left(s_{i+1}, s_i\right)$ are
  not saturated;
\item for every $k$, $i < k < j$, if $s_k \in I_P$ (resp. $s_k \in J_P$)
  then $\left(s, s_k\right)$ (resp. $\left(s_k,t \right)$) is saturated.
\end{itemize}
This might seem like an elaborate notion, but in fact, it only translates
whether, in terms of flow increase, we can apply our transformation.

We are going to prove that the following algorithm builds a maximum flow:
\begin{enumerate}
\item we begin with a null flow $f$;
\item while there exists two vertices $s_i$ and $s_j$ with $i < j$ which
  are $f$-tractable:
  \begin{enumerate}
  \item we take the first pair of vertices in lexicographic order which is
    $f$-tractable,
  \item and we augment the flow $f$ by $1$ on the arcs belonging to the path $[s, s_i, s_{i+1}, \ldots, s_j, t]$ or $[s, s_j, s_{j-1},
    \ldots, s_i, t]$ (depending on the dominant color of $s_i$);
  \end{enumerate}
\item we return the flow $f$.
\end{enumerate}
\begin{proof}[Correctness of the greedy algorithm]
Let $(s_i,s_{i+1})$ be the first pair of vertices in lexicographic order which, at the beginning of the algorithm, is $f$-tractable. We suppose here that $s_i$ corresponds to a black dominant column (the other case is symetrical). We want to prove that there is a maximum flow $f$, which minimizes $\sum_{v\in T_P} f(v)$, such that none of the following values are null: $f((s,s_i)),f((s_i,s_{i+1})),f((s_{i+1},t))$. If a maximum flow following this last condition does not exist, then one of those three arcs must be null and another one must be saturated, otherwise the flow $f$ can be augmented. Now, by the minimality of $\sum_{v\in T_P} f(v)$, it follows that $(s,s_i)$ is not saturated (otherwise we can remove 1 to $f$ on every arc belonging to a minimal path from $s$ to $t$ using $(s,s_i)$ with non-zero flow, and add 1 to $f((s,s_i)),f((s_i,s_{i+1})),f((s_{i+1},t))$. We obtain a maximum flow which contradicts the minimality of $\sum_{v\in T_P} f(v)$.) The arc $(s_i,s_{i+1})$ cannot be saturated for the same reason. So, $(s_{i+1},t)$ is saturated. But in this case, as $(s,s_i)$ and  $(s_i,s_{i+1})$ are not saturated, we can remove 1 to $f$ on every arc belonging to a non-zero flow path from $s$ to $t$ using $(s_{i+1},t)$, and add 1 to $f((s,s_i)),f((s_i,s_{i+1})),f((s_{i+1},t))$. We again obtain a new maximum flow with total weigth at most $\sum_{v\in T_P} f(v)$. That proves what we want. Now, the next recursive steps of the algorithm are exactly proved by the same way but on transport networks obtained by convenient diminution of the arc capacities.
\end{proof}

To briefly but correctly analyze the complexity of this algorithm, we have to describe the representation of the entry. If the Manhattan polyomino is given by the list of its column heights, it can be stored in $O(l\log(\max(h_i)))$ bits where $l$ is the number of columns. We can consider two steps. The first one consists in obtaining the transport network $T_P$ from $P$. This can be done in one pass and we make $O(l)$ elementary operations ($\min$, division by 2).  The second part consists in solving the transport networks problem. Using stacks, the algorithm can solve the problem making a unique pass on the vertices of $T_P$. So, its complexity is in $O(|I_P|)$. Thus, globally the algorithm works in $O(l)$ elementary operations. If we make the hypothesis that the number of columns is approximatively the average of the height of the columns. Our algorithm is clearly sublinear in $n$ the number of unit squares. Moreover, if we consider the complexity according to the size of the input, this algorithm is linear.

\section{Conclusion}
In this paper, we have illustrated that particular optimal partial domino
tiling problems can be solved with specific algorithms. We have not yet
been able to evidence a linear algorithm for the general problem of
partial domino tiling of polyominoes without holes. This appears to be an
attainable goal, as \cite{Th} has devised such an algorithm for the exact
domino tiling problem.

Currently, the best known complexity for the general partial domino tiling
problem is $O(n\cdot\sqrt{n})$ using a \emph{maximum matching algorithm in
  bipartite graphs}. To improve this bound, it seems irrelevant to obtain
an analogue to the notion of height which is a cornerstone of the tiling
algorithms. Nevertheless, it can be proven that two optimal partial
tilings are mutually accessible using extended flips (which contain
classical flips \cite{Bo2,Re} and taquin-like moves). This fact involves
in a sense a weak form of ``lattice structure'' over the set of optimal
partial tilings. This is relevant and could be taken into account to
optimize the algorithms.

\newpage

\section{Annexe (proof of theorem 1)}

So, we have proved that on the sequence of Manhattan polyominoes obtained by successive planing transformations, the function
$|I_\_|+|J_\_|-2v(\_)$ is invariant. 
We need another invariant to conclude that $|I_P|+|J_P|-2v(P) \leq d(P).$ It follows from a list of ``bottlenecks'' which prevent the matching of columns. 

We denote by $P_{\mathrm{red}}=(h^{\mathrm{red}}_1,...,h^{\mathrm{red}}_n)$ the Manhattan polyomino obtained from $P$ after $v(P)$
planing transformations.

We can suppose that the least $x$ such that $(x,x')$ belongs to $Y_P$ (the set of edges of $F_P$ that crosses the minimal cut)
 corresponds to a black dominant column (if this is not the case, we can inverse the colors). Let $g(x,x')$ be the smallest $k
\in \set{x, \ldots, x'}$ such that $h_{k} = X(x, x')$. We denote by $\mathcal{G}=\{g(x,x'); (x,x')\in Y_P\}$ the \emph{list of bottlenecks} of $P$ and we sort the elements of $\mathcal{G}$ in ascending order.

$\mathcal{G}=\left (g^0_1,\ldots,g^0_k \right )$. Now, if $i$ is odd (resp. even) and is the index of a white dominant column, we put $g_i=g^0_i-1$ (resp. $g_i=g^0_i+1$) otherwise we put $g_i=g^0_i$. With the hypothesis, we can observe that $(h_1,h_2,\ldots,h_{g_1})$ has more black unit squares than white ones. 
We put $$H_P=B\left(\left(h_1,h_2,\ldots,h_{g_1}\right)\right)\cup
B\left(\left(h_{g_2},\ldots,h_{g_3}\right)\right)\cup\ldots$$ and
$$K_P=W\left(\left(h_1,h_2,\ldots,h_{g_1}\right)\right)\cup
W\left(\left(h_{g_2},\ldots,h_{g_3}\right)\right)\cup\ldots$$ 
We have $2(|H_P|-|\Gamma(H_P)|)\leq d(P)$ and we want to prove that $2(|H_P|-|\Gamma(H_P)|)= |I_P|+|J_P|-2v(P)$. In order to do that, let $P'$ be the Manhattan polyomino obtained from $P$ by a planing transformation
as we have processed previously. Firstly, we are going to show that $|H_P|-|\Gamma(H_P)|=|H_{P'}|-|\Gamma(H_{P'})|$ where
$$H_{P'}=B\left(\left(h'_1,h'_2,...,h'_{g_1}\right)\right)\cup
B\left(\left(h'_{g_2},...,h'_{g_3}\right)\right)\cup...$$ We put $\mathcal{B}(H_P)$ (resp. $\mathcal{W}(H_P)$) the number
of odd black dominant columns (resp. white) in $P$ which belongs to the columns of indices in
$E=\{1,2,...g_1\}\cup\{g_2,...,g_3\}\cup...$. An easy observation allows us to see that

\begin{eqnarray}
\nonumber |H_{P}|-|\Gamma(H_{P})| & = &|H_{P}|-|K_{P}|-\sum\limits_{v\in Y_P}{c(v)}\\
\nonumber & = &\mathcal{B}(H_P)-\mathcal{W}(H_P)-\sum\limits_{v\in Y_P}{c(v)}
\end{eqnarray}
When we apply the planing transformation on $P$, we delete :
\begin{itemize}
 \item a black dominant column in $E$ and a white column in $E^c$ and $\sum\limits_{v\in Y_{P'}}{c(v)}=\sum\limits_{v\in Y_P}{c(v)}-1$,
\item a black dominant and a white dominant column in $E$ or a black dominant and a white column in $E^c$ and $\sum\limits_{v\in Y_{P'}}{c(v)}=\sum\limits_{v\in Y_P}{c(v)}$.
\end{itemize}
  This fact proves that $|H_{\_}|-|\Gamma(H_{\_})|$ is invariant on every sequence of Manhattan polyominos obtained by successive planing transformations.

Now, let us remark that 
$$|H_{P_{\mathrm{red}}}|-|\Gamma(H_{P_{\mathrm{red}}})|=\mathcal{B}(H_{P_{\mathrm{red}}})=|I_{P_{\mathrm{red}}}|$$ ($\mathcal{W}(H_{P_{\mathrm{red}}})=\sum\limits_{v\in Y_{P_{\mathrm{red}}}}{c(v)}=0)$ and that $|I_{P_{\mathrm{red}}}|=|I_P|-v(P)$. Indeed, we have made $v(P)$ planing transformations and each of them reduces by 1 the number of black dominant columns. 

So, $|H_{P}|-|\Gamma(H_{P})| =|H_{P_{\mathrm{red}}}|-|\Gamma(H_{P_{\mathrm{red}}})|=|I_P|-v(P)$.

Moreover, we have assumed initially that the polyomino is balanced. So, $|I_P|=|J_P|$. Consequently,
we have proved that $2(|H_{P}|-|\Gamma(H_{P})|)=|I_P|+|J_P|-2v(P)$. 
Thus, we can conclude that $|I_P|+|J_P|-2v(P)\leq d(P).$

\end{document}